\documentclass[12pt]{article}

\usepackage{amsthm}
\usepackage{amsmath}
\usepackage{amssymb}
\usepackage{enumerate}
\usepackage{xcolor}
\usepackage{listings}
\usepackage{multirow}
\usepackage{color}
\usepackage[pdfencoding=auto, psdextra]{hyperref}
\usepackage{doi}
\usepackage{booktabs}
\usepackage{rotating}
\usepackage{bm}
\usepackage{tablefootnote}
\usepackage[normalem]{ulem}
\usepackage{array}
\usepackage{longtable}
\usepackage{bbm}
\usepackage{cleveref}
\usepackage{tabularx}

\newcommand{\E}{\mathbb{E}}
\newcommand{\ed}{\mathrm{d}}

\renewcommand{\P}{\mathbb{P}}

\newcolumntype{L}[1]{>{\raggedright\let\newline\\\arraybackslash\hspace{0pt}}p{#1}}

\DeclareMathOperator{\argmax}

\newcommand{\drift}{\mathrm{drift}}
\newcommand{\vol}{\mathrm{vol}}

\newtheorem{theorem}{Theorem}

\newtheorem{proposition}[theorem]{Proposition}

\theoremstyle{definition}

\numberwithin{equation}{section}
\numberwithin{theorem}{section}

\begin{document}

\author{\small{John Armstrong\textsuperscript{a}, James Dalby\textsuperscript{b}}\\
\small{\textsuperscript{a,b}Department of Mathematics, King's College London, UK}\\
\small{\textsuperscript{a,*}john.armstrong@kcl.ac.uk, \textsuperscript{b}james.dalby@kcl.ac.uk} }

\title{Optimal mutual insurance against systematic longevity risk}

\date{}              


\maketitle

\begin{abstract}
    We mathematically demonstrate how and what it means for two collective pension funds to mutually insure one another against systematic longevity risk. The key equation that facilitates the exchange of insurance is a market clearing condition. This enables an insurance market to be established even if the two funds face the same mortality risk, so long as they have different risk preferences.
    Provided the preferences of the two funds are not too dissimilar, insurance provides little benefit, implying the base scheme is effectively optimal. When preferences vary significantly, insurance can be beneficial. 
\end{abstract}

\renewcommand{\thefootnote}{}
\footnotetext{* Corresponding author. Strand,
London,
WC2R 2LS, UK. john.armstrong@kcl.ac.uk. +44 (0)20 78482855}

\section{Introduction}
In this paper, we calculate the optimal consumption, investment and insurance purchase strategies for a collective pension fund in a Black--Scholes market, subject to
both idiosyncratic and systematic longevity risk. We call the collective fund we study an ``Insured Drawdown Scheme'', which is characterised by a tontine structure/longevity credit system. This allows the fund to insure against idiosyncratic longevity risk. However, undiversifiable systematic longevity risk, is more difficult to insure against through a simple tontine. Therefore, additionally, a fund could trade mutual insurance contracts with another fund to protect against this risk. 

Armstrong et al. \cite{armstrong_buescu_dalby} studies the optimal consumption-investment problem in a Black-Scholes Market for a single insured drawdown fund under the effects of idiosyncratic and systematic longevity risk. 
The present paper builds on this by considering two insured drawdown funds with the same mortality risk and investment opportunities, along with the inclusion of a systematic mortality risk insurance market where the two funds can trade with each other. According to the theory of comparative advantage, when the two funds have different risk preferences, they value mortality risk differently and there is therefore an incentive for them to trade in this market for mutual benefits. To the best of our knowledge, two collective funds exchanging mutual insurance contracts on longevity risk via an internal market has not been considered in the literature.

We consider the decumulation stage of retirement and focus on tractable problems that yield equations on which analytical and numerical progress can be made. 
For this reason, throughout, we model the preferences of the funds using power utility, or more generally, Epstein--Zin preferences.
Besides their tractability (see \cite{campbellViceira} for a study of this), Epstein-Zin preferences have the advantage of separating risk-aversion from the diminishing returns of increased consumption at a moment in time. This has allowed the resolution of various asset pricing problems \cite{bansalYaron, bansal, benzoniEtAl, bhamraEtAl}. We therefore model our two collective funds as a continuous time stochastic optimal control problem where we seek to maximise Epstein-Zin utility with mortality. 

We work in continuous time so that the insurance market is complete. This ensures that the agents in our model are able to replicate any possible insurance contract. Thus our optimal investment strategies are optimal among all possible insurance contracts. This approach is similar in spirit to Cui and Ponds \cite{cui_ponds}, who consider a wage related swaps market. However, they work in discrete time which limits them to an incomplete market where a specific contract must be renegotiated each year and this is not necessarily optimal.

We derive the most general Hamilton-Jacobi-Bellman (HJB) equations which have three types of controls, the consumption rate, the investment quantity in a risky asset and the insurance purchase rate. In total, we have three equations to solve, two coupled three-dimensional partial differential equations (PDEs) resulting from the two HJB equations and an insurance market clearing condition. This a difficult task, so we focus on a simple sub-case of this most general problem, in which one of the funds is finite and the other infinite. This reduces the dimension of the PDEs from three to one, making them solvable, whilst still allowing us to quantify the maximal possible benefit insurance can provide. 

We study the finite-infinite fund problem under two different mortality models: (i) a stylised mortality model and (ii) a more realistic continuous time analogue of the Cairns--Blake--Dowd (CBD) mortality model \cite{CBD_article}. The stylised model introduces a time symmetry allowing us to reduce the dimension of the HJB equation and analytically determine the value function, optimal strategies and insurance price. This is beneficial as it provides a concrete example to highlight the novel features of operating co-dependent collective funds. 
The CBD model does not posses such a time symmetry, so the resulting HJB equation is solved numerically using the Crank-Nicholson scheme. This enables us to provide realistic values on the potential benefits of mutual insurance. For most combinations of preferences, the maximum benefit uplift experienced by a fund is typically less than 6\%. We interpret this to mean the additional complexity introduced from insurance is not a worthwhile exercise when the improvement in outcomes is small. In some scenarios insurance may be beneficial.


Our work contributes to the well established literature on optimal investment. Problems related to our setup include the seminal Merton problem \cite{merton1969lifetime} and optimal investment problems under Epstein-Zin utility without mortality \cite{xing, optimal_EZ_finite, optimal_EZ_infinite}. It also adds to the growing literature on tontines. See \cite{milevsky} for a review of the history of tontines and the most recent literature. Specific works of interest include Milevsky and Salisbury \cite{milevsky2015}, who consider optimal investment for
a tontine invested in a bond only; Chen and Rach \cite{chenAndRach} consider a tontine with a minimum payout; Chen et al. \cite{chenRachSehner2020} consider a combination of tontines and annuities; while Boado-Penas et al. \cite{boadoPenas} consider a system that keeps funds within certain corridors or limits. 
More generally, this work enriches the discussion on the potential benefits of risk sharing in collective based pension schemes \cite{gollier2008, cui, branger}, but these approaches, based on central planning, require compulsory membership \cite{bovenberg2007}, while our insured drawdown scheme does not. 

In \Cref{sec:two_fund}, we introduce the defining equations for the optimal operation of two finite collective funds with access to insurance contracts on mortality. In \Cref{sec:finite_infinite}, we introduce the specific case of a finite and an infinite fund and study this under different mortality models. In \Cref{sec:conclusions}, some conclusions are presented.

\section{Mutual insurance for two funds}\label{sec:two_fund}
We begin by developing the mathematical framework for two collective funds to mutually insure one another. The dynamics of the problem are specified by a total of 6 equations. Indexing the two funds by $i=1,2$, we have two wealth processes $w_i$ and two equations for their dynamics, a mortality-rate process  $\lambda$ and its dynamics, a single risky asset of price $S$ and its dynamics, and an equation governing the size of each fund $n_i$. These are
\begin{subequations}
\begin{align}
    &\ed w_i = (\lambda w_i + r(w_i-q^a_i S)-q^c_i p-c_i+q^a_i\mu S +q^c_i \drift_{\lambda} )\ed t \nonumber\\
    &\quad\quad+ q^a_i\sigma S \ed W^1+q_i^c \vol_{\lambda} \ed W^2,\label{eq:wealth}\\
    &\ed S = \mu S \ed t + \sigma S \ed W^1,\label{eq:BSM}\\
    &\ed\lambda = \drift_{\lambda} \ed t + \vol_{\lambda} \ed W^2,\label{eq:mortality_sde}\\
    & \ed n_i = -\lambda n_i \ed t.
\end{align}
\end{subequations}
Each fund has three controls, $c_i$, the rate of consumption, $q_i^a$, the quantity of the risky asset purchased and $q_i^c$, the rate at which insurance contracts of price $p$, are purchased. These insurance contracts represent insurance against increases in $W^2$ and hence against increases in the mortality rate. Our market model is given by \eqref{eq:BSM} i.e.\ the Black--Scholes model with drift $\mu$, volatility $\sigma$ and $W^1$ is a Brownian motion. As the second term in equation \eqref{eq:wealth} indicates, we are also assuming
that their is a risk-free asset with constant interest rate $r$.
The force of mortality is given by \eqref{eq:mortality_sde}, where $W^2$ is a Brownian motion independent of $W^1$. We take the mortality rate to be the value of the payout of insurance contracts at each time. Our collective fund is characterised by a tontine/longevity credit system whereby the funds of deceased members are shared evenly among survivors, hence the $\lambda w_i$ term in \eqref{eq:wealth}. In the absence of insurance contracts the dynamics of each fund will evolve independently, hence $p$ is what couples the two funds and reflects information about the ratio of their size's and wealth's. We call the decoupled problem without insurance ``the one-fund problem".

It is important to note that by allowing investors to purchase this insurance product, combined with the fact that we are working in a continuous time model, we have created a complete market in systematic-longevity risk. Our model allows the parties to replicate arbitrary longevity-risk contracts and so there is no loss of generality in using this single contract.

In order to make the exchange of insurance fair, the price of insurance contracts should be set endogenously. This is done via the introduction of one more equation, a market clearing condition. This equation sets the price of systematic-longevity-risk contracts to be the price such that the internal insurance market clears when everyone behaves optimally, that is
\begin{equation}
    n_1 q^{c*}_1 + n_2 q^{c*}_2 = 0,\label{eq:clearing}
\end{equation}
where $q^{c*}_i$ denotes the optimal insurance contract purchase rate.
This is the defining equation for operating two collective funds with insurance fairly.

The market clearing condition \eqref{eq:clearing} allows us to determine the pricing measure for arbitrary systemic-longevity-risk contracts (see \eqref{eq:general_price}).  
By the Martingale representation theorem, any derivative contract can be replicated by trading in the continuous time insurance contract and its price will be determined by this measure.

We assume each fund seeks to maximize a value function given by homogeneous Epstein--Zin preferences with mortality. We work with these preferences as they yield analytically tractable problems. Epstein--Zin preferences are understood most easily in discrete time and we refer the reader to \cite{armstrong_buescu_dalby} for a discussion of this. Epstein--Zin preferences can also be defined in continuous time as the solution of a backwards stochastic differential equation: 
\begin{equation}
d V_t = -f(c_t, V_t, \lambda_t) \, \ed t + Z_t \ed W_t,
\label{eq:ezBSDE}
\end{equation}
where $f$ is the Epstein--Zin aggregator,
\begin{equation}
    f(c,V,\lambda):=\frac{1}{\rho}c^{\rho} \left( \alpha V \right)^{1-\frac{\rho}{\alpha}} - \left(\frac{\alpha}{\rho}\delta+\lambda\right) V \textrm{ for }i=1,2.\label{eq:general_EZ_aggregator2}
\end{equation}
$0\neq\alpha<1$ is a monetary risk aversion parameter, $0\neq\rho<1$ is a satiation parameter and $\delta>0$ is a discount factor. We set $\delta=0$, since we want to focus on the discounting effects of mortality. We note that in the case $\alpha=\rho$, the aggregator corresponds to von Neumann--Morgenstern utility (with mortality).
A heuristic derivation of continuous time Epstein--Zin preferences with mortality can be found in Appendix \ref{sec:ezAggregatorMotivation}. The argument follows the same steps as those for standard Epstein--Zin preferences with no mortality \cite{optimal_EZ_finite}.

Our objective is to find controls $(c^i,q^a_i,q^c_i)$, such that
\begin{multline}
V_i(t,S,w_1,w_2,n_1,n_2,\lambda)=\\
\underset{(c^i,q^a_i,q^c_i)\in\mathcal{A}(w_0)}{\sup}\mathbb{E}\left[\int_t^\infty  f_i(c_i,V_i,\lambda) ds\right],\textrm{ for }i=1,2,
\end{multline}
where $\mathcal{A}(w_0)$ is the set of strategies starting from initial wealth $w_0$, such that the wealth process remains non-negative for all times. Assuming each $V_i$ is smooth, an It\^o expansion yields the corresponding SDEs.
The martingale principal of optimal control (see \cite{Rogers}) then yields the HJB equation for each fund:
\begin{align}
    &0=\underset{(c^i,q^a_i,q^c_i)\in\mathcal{A}(w_0)}{\sup}\left[\frac{\partial V_i}{\partial t}+\right.\nonumber\\
    &\left.\frac{1}{2}\sum^2_{j=1}\frac{\partial^2 V_i}{\partial w_j^2} \left((q^a_j\sigma S)^2 + (q_j^c \vol_{\lambda})^2\right) + \frac{\partial^2 V_i}{\partial w_1\partial w_2}(q^a_1 q^a_2(\sigma S)^2 + q_1^c q_2^c (\vol_{\lambda})^2) + \nonumber\right.\\
    &\left.\sum^2_{j=1}\frac{\partial V_i}{\partial w_j} \left(\lambda w_j + r(w_j-q^a_j S)-q^c_j p-c_j+q^a_j\mu S +q^c_j \drift_{\lambda}\right)+\nonumber\right.\\
    & \left.\drift_{\lambda} \frac{\partial V_i}{\partial\lambda }+\frac{1}{2}(\vol_{\lambda})^2 \frac{\partial^2 V_i}{\partial\lambda^2 }+\sum^2_{j=1}\frac{\partial^2 V_i}{\partial\lambda \partial w_j}(\vol_{\lambda})^2q^a_j
   -\sum^2_{j=1}\frac{\partial V_i}{\partial n_j} \lambda  n_j +\nonumber\right.\\
   &\left.\sum^2_{j=1} \frac{\partial^2 V_i}{\partial S \partial w_j}q^a_j (\sigma S)^2 +\frac{\partial V_i}{\partial S}\mu  S +\frac{1}{2} \frac{\partial^2 V_i}{\partial S^2}(\sigma  S)^2 +f_i(c_i,V_i,\lambda)\right]
   \text{ for $i=1,2$}.
   \label{eq:two_fund_hjb}
\end{align}
These two HJB equations are coupled.

Solving this set of equations is a difficult task. An analytical solution of this problem appears highly unlikely, whilst obtaining numerical solutions also faces various problems. 
Equation \eqref{eq:two_fund_hjb} could be solved with policy iteration to compute the optimal strategies (see \cite{Policy_iteration}), but this is computationally expensive for a three-dimensional PDE. One may also try to exploit symmetries of the problem to simplify the equations. For example, we can ignore the fund sizes ($n_1,n_2$) as variables, since a change in them should have the same effect as a change in wealth for a fixed number of individuals. Furthermore, the ratio of the fund sizes will remain the same for all time since both funds obey the same mortality model. 
This suggests we express the equation in term of the wealth ratio $w_2/w_1$, so we obtain a two-dimensional problem. However, the resulting equations are still complex, highly non-linear and likely to run into numerical issues. 

Given these issues, and that we do not know a priori if much benefit can even be gained from exchanging mutual insurance contracts on mortality, we simplify the problem by determining the maximum possible benefit. This is achieved by considering the case where one of the funds is effectively infinite in terms of its size and wealth. This is therefore a limiting case of the setup described above i.e.\ $n_1,w_1\to\infty$, or $n_2,w_2\to\infty$. The infinite fund sets the price of the insurance contracts and matches all demand for these contracts from the small fund without altering its position. Hence, this situation yields the maximum possible benefit a fund can achieve from buying/selling insurance with another fund. By comparing our results to a single fund that does not purchase insurance (the one-fund problem), we can determine how much additional benefit insurance can give.

\section{A finite and an infinite fund}\label{sec:finite_infinite}

Without loss of generality, we assume the first fund indexed by $i=1$, is finite, and the second fund indexed by $i=2$, is infinite. In this situation we have $V_1(t,S,w_1,w_2,n_1,n_2,\lambda)=V_1(t,S,w_1,n_1,\lambda)$, since $w_2$ and $n_2$ will not change and $V_2(t,S,w_1,w_2,n_1,n_2,\lambda)=V_2(t,S,w_2,n_2,\lambda)$, since fund one will not influence fund two. 

Our problem has three symmetries. Firstly, at any time $t$, the Black-Scholes market with initial stock price $S_t$, is equivalent to the same market with any other possible initial stock price. As a result, one expects that the optimal investment strategy will be independent of the stock price $S_t$. Secondly, $V$ is positively homogeneous of order $\alpha$ i.e.\ 
$V(\zeta c,\lambda)=\zeta^\alpha V(c, \lambda)$. Since we are in the Black-Scholes market, one expects that the value function will also be positively homogeneous of order $\alpha$ in the wealth. Third, a change in the fund size will be the same as a change in wealth for a fixed number of individuals.
This motivates an ansatz for the HJB equation.

\begin{proposition}
If one substitutes
\begin{equation}
V_1(t,\lambda,w_1)=w_1^{\alpha_1} g_1(\lambda,t), \text{ and } V_2(t,\lambda,w_2)=w_2^{\alpha_2} g_2(\lambda,t),\label{eq:two_fund_ansatz}
\end{equation}
into the HJB equation \eqref{eq:two_fund_hjb}, for the case of a finite and infinite fund, 
one obtains the partial differential equation 
\begin{align}
    &(\alpha_1-1) \left(2 \drift_\lambda \frac{\partial g_1}{\partial\lambda}+2 \frac{\partial g_1}{\partial t}+\vol_\lambda^2 \frac{\partial^2 g_1}{\partial\lambda^2}\right)+\frac{2 \alpha_1 \vol_\lambda^2 }{g_2}\frac{\partial g_2}{\partial\lambda} \frac{\partial g_1}{\partial\lambda}\nonumber\\
    &-\frac{\alpha_1 \vol_\lambda^2
   }{g_1}\left(\frac{\partial g_1}{\partial\lambda}\right)^2\nonumber\\
   &+g_1 \left(\frac{2 (\alpha_1-1)
   \left(\alpha_1 \left(\rho_1 (r+\lambda) +(1-\rho_1) (\alpha_1 g_1)^{\frac{\rho_1}{\alpha_1(\rho_1-1)}}\right)-\lambda  \rho_1 \right)}{\rho_1}\right.\nonumber\\
   &\left.-\frac{\alpha_1 \vol_\lambda^2
   }{g_2^2}\left(\frac{\partial g_2}{\partial\lambda}\right)^2-\frac{\alpha_1 (r-\mu )^2}{\sigma ^2}\right)=0,\label{eq:inf_fund_optimal_hjb}
\end{align}
for the finite fund and
\begin{align}
    &\alpha_2 g_2\Bigg(  -\frac{\mu ^2}{2(\alpha_2 -1) \sigma ^2} + (\alpha_2  g_2)^{\frac{\rho_2 }{(\rho_2-1)\alpha_2 }}\left(\frac{1}{\rho_2 }-1\right)+\nonumber\\
   & \lambda  \left(1 -\frac{1}{\alpha_2}\right)-\frac{r^2}{2(\alpha_2 -1) \sigma ^2}+ r
   \left(\frac{\mu }{\left(\alpha_2 -1\right) \sigma ^2}+1\right)\Bigg)+\nonumber\\
    &\frac{\partial g_2}{\partial t} + \drift_\lambda \frac{\partial g_2}{\partial \lambda}+\frac{1}{2} \vol^2_\lambda \frac{\partial^2 g_2}{\partial \lambda^2}=0\label{eq:g_eqtn},
\end{align}
for the infinite fund.
The optimal consumption rate, investment quantity and insurance purchase rate are:
\begin{subequations}\label{eq:finite_strats}
\begin{align}
    & c_1^* = w_1 (\alpha_1 g_1)^{\frac{\rho_1}{\alpha_1(\rho_1-1)}}, \\
    & q_1^{a*} = \frac{w_1 (r-\mu )}{(\alpha_1-1) \sigma ^2 X},\\
    & q_1^{c*} = \frac{w_1 \left(\frac{\partial g_2}{\partial\lambda} g_1-g_2 \frac{\partial g_1}{\partial \lambda}\right)}{(\alpha_1-1) g_2 g_1}.\label{eq:insurance_strat}
\end{align}
\end{subequations}
Note, $q_1^{a*}$ is the investment strategy from the classical Merton problem. The price the infinite fund charges the small fund for insurance is
\begin{equation}
    p = \drift_\lambda + \frac{\vol^2_\lambda}{g_2} \frac{\partial g_2}{\partial \lambda}.\label{eq:inf_price}
\end{equation}.

\end{proposition}

\begin{proof}
The key step is determining the price $p$, that the infinite fund charges. The correct price is the one for which the infinite fund does not wish to actively trade. This can be determined by considering the case when fund one and two have the same preferences (i.e.\ $\alpha_1=\alpha_2, \rho_1=\rho_2$). In this scenario there is no reason for either fund to trade and the no trade price can be found. By substituting the infinite fund's preferences into this no trade price, we obtain its no trade price.

Differentiating \eqref{eq:two_fund_hjb} with respect to $q^c_1$ and $q^c_2$, 
setting the results equal to zero and solving for $q^c_1$ and $q^c_2$, we can find the optimal rates.
The market clearing condition \eqref{eq:clearing} can then be solved for the price $p$, which is 
\begin{align}
    p = &\left[\drift_\lambda \frac{\partial V_2}{\partial w_2} \left(\frac{\partial^2 V_1}{\partial w^2_1}-\frac{n_2}{n_1} \frac{\partial^2 V_1}{\partial w_1\partial w_2}\right)+\frac{n_2}{n_1} \drift_\lambda  \frac{\partial V_1}{\partial w_1}\frac{\partial^2 V_2}{\partial w^2_2}-\right.\nonumber\\
    &\left.\frac{n_2}{n_1}
   \vol^2_\lambda \frac{\partial^2 V_1}{\partial w_1\partial w_2} \frac{\partial^2 V_2}{\partial \lambda\partial w_2}+\frac{n_2}{n_1} \vol^2_\lambda \frac{\partial^2 V_2}{\partial w^2_2} \frac{\partial^2 V_1}{\partial \lambda\partial w_1}-\drift_\lambda \frac{\partial V_1}{\partial w_1}
   \frac{\partial^2 V_1}{\partial w_1\partial w_2}\right.\nonumber\\
   &+\left.\vol^2_\lambda \frac{\partial^2 V_1}{\partial w^2_1} \frac{\partial^2 V_2}{\partial \lambda\partial w_2}-\vol^2_\lambda \frac{\partial^2 V_2}{\partial w_1\partial w_2} \frac{\partial^2 V_1}{\partial \lambda\partial w_1}\right]\Bigg/\nonumber \\
   &\left[\frac{\partial V_2}{\partial w_2}
   \left(\frac{\partial^2 V_1}{\partial w^2_1}-\frac{n_2}{n_1} \frac{\partial^2 V_1}{\partial w_1\partial w_2}\right)+\frac{\partial V_2}{\partial w_2} \left(\frac{n_2}{n_1} \frac{\partial^2 V_2}{\partial w^2_2}-\frac{\partial^2 V_2}{\partial w_1\partial w_2}\right)\right].
    \label{eq:general_price}
\end{align}
When preferences are equal, the problem for the two funds decouple (since there is no trade), meaning the value functions should be the same up to some scaling determined by wealth, that is
\begin{equation}
    V_1=w_1^\alpha g(\lambda,t),\quad V_2=w_2^\alpha g(\lambda,t).
\end{equation}
Substituting these ans\"atze into the price found using the preceding steps, gives the following no-trade price
\begin{equation}
    p = \drift_\lambda + \frac{\vol^2_\lambda}{g} \frac{\partial g}{\partial \lambda}.\label{eq:no_trade_price}
\end{equation}
 Replacing $g$ in \eqref{eq:no_trade_price}, with the solution of the one-fund problem for the infinite fund, we obtain \eqref{eq:inf_price}.


We can now optimise the finite fund subject to this insurance price.
Substituting our ansatz \eqref{eq:two_fund_ansatz} and the insurance price \eqref{eq:no_trade_price} into (2.6) with $i=1$,
and differentiating with respect to $c_1,q_1^a,q_1^c$, we find the optimal strategies in \eqref{eq:finite_strats}. Substitution of these strategies into the HJB equation, yields the PDE \eqref{eq:inf_fund_optimal_hjb}. We can repeat the same steps for the infinite fund to get the PDE \eqref{eq:g_eqtn}. The strategies the infinite fund follows are unimportant (since we wish to study the finite fund) beyond the fact $q_2^{c*}=0$, when the price is given by \eqref{eq:inf_price}.

\end{proof}

\subsection{A stylised mortality model}
We now solve \eqref{eq:inf_fund_optimal_hjb} analytically through a special choice of mortality model. We introduce the stylised model considered in \cite{armstrong_buescu_dalby}, namely
\begin{equation}
    \ed \lambda= a \lambda^2 \, \ed t + b \lambda^{\frac{3}{2}} \, \ed W,\label{eq:stylised_mortality_SDE}
\end{equation}
for some constants $a$ and $b$. This model does not match human mortality data particularly well, but it does have similar qualitative properties: $\lambda_t$ will always be positive, it will explode to $+\infty$ in a finite time, ensuring that all members of the funds die in finite time and,
ignoring short-term fluctuations, mortality rates increase with age.
The advantage of this model is that \eqref{eq:stylised_mortality_SDE} is time scale-invariant, thereby introducing a symmetry in time and removing it as a variable in the case $r=\mu=0$. 
This suggests solutions should take the form
\begin{equation}
    g_1(\lambda) = B_1 \lambda^\xi,\text{ and }g_2(\lambda) = B_2 \lambda^\xi.\label{eq:style_ansatz}
\end{equation}

\begin{theorem}
The HJB equation for a finite and infinite fund, under the mortality model \eqref{eq:stylised_mortality_SDE}, with $r=\mu=0$, has trivial solutions and the analytical solutions 
\begin{subequations}
    \begin{align}
    &V_1= w^{\alpha_1} \frac{(A(a,b, \alpha_1,\alpha_2,\rho_1,\rho_2)\lambda) ^{\frac{\alpha_1 (\rho_1-1)}{\rho_1}}}{\alpha_1},\label{eq:two_fund_analytical}\\
    & 
    V_2=w_2^{\alpha_2} \frac{ \left(\left(\frac{(\alpha_2 -1) \rho_2 }{\alpha_2  (\rho_2 -1)}+a+\frac{b^2 (\alpha_2  (\rho_2 -1)-\rho_2 )}{2 \rho_2
   }\right)\lambda\right)^{\frac{\alpha_2  (\rho_2 -1)}{\rho_2 }}}{\alpha_2 }\label{eq:analytical},
\end{align} 
\end{subequations}
where
\begin{multline}
    A(a,b, \alpha_1,\alpha_2,\rho_1,\rho_2) = \\
    \frac{a (\rho_1-1)+\frac{(\alpha_1-1) \rho_1}{\alpha_1}-\frac{b^2 \left(\alpha_2^2 \rho_1^2 (\rho_2-1)^2-2 \alpha_1 \alpha_2 (\rho_1-1) \rho_1 \rho_2 (\rho_2-1)+(\rho_1-1) \rho_2^2 (\alpha_1 (2 \rho_1-1)-\rho_1)\right)}{2(\alpha_1-1) \rho_1 \rho_2^2}}{\rho_1-1},
\end{multline}
as long as $A$ and the bracketed term in \eqref{eq:analytical} are positive. 
\end{theorem}
The proof is a calculation using the ansatz \eqref{eq:style_ansatz}.

When the optimization problem is ill-posed, the supremum of the value function may be $-\infty$, $0$ or $\infty$. This explains why
we sometimes obtain complex solutions to the HJB equation. By evaluating the value-function for other strategies that fit the ansatz, but which do give real values, one can analyze the supremum of the value function in those cases where the supremum is not attained. See \cite{armstrong_buescu_dalby} for an analysis of this for a one-fund problem.

A fund should always be better off with insurance than without. That is, $V_I:=w_I^{\alpha} g_I>V_{NI}:=w_{NI}^{\alpha}g_{NI}$, where $V_I$ denotes the value function with insurance and $V_{NI}$ the value function without insurance i.e.\ the solution to the one-fund/infinite-fund problem \eqref{eq:analytical}. In order for the problem without insurance, to have the same value function as the problem with insurance (at time zero), we need to increase the initial wealth by a certain percentage relative to the associated problem with insurance. We will call this percentage the \emph{``insurance benefit''} of the finite-infinite fund problem. This percentage also defines the \emph{``maximum insurance benefit"} the finite fund could obtain from trading with the other fund in the case that it too is finite.

In \Cref{fig:differences_two_fund}, for the case of von Neumann--Morgenstern preferences, we plot the insurance benefit under our stylised model to begin to illustrate the benefit insurance contracts can provide. Given the stylised nature of our mortality model, the values in this plot are not particularly realistic. However, we expect this does provide an initial upper bound on the benefit one could expect to see when using more realistic mortality models. This suggest insurance will typically have less than a 10\% benefit. 

The stylised model also reveals some useful facts about the behaviour of the finite fund. Substituting the solution \eqref{eq:two_fund_analytical} into \eqref{eq:insurance_strat}, reveals the optimal insurance purchase rate is,
\begin{equation}
    q_1^{c*}=\frac{w}{\lambda(\alpha_1-1)}\left(\alpha_2-\alpha_1 + \frac{\alpha_1}{\rho_1} - \frac{\alpha_2}{\rho_2}\right).\label{eq:style_rate}
\end{equation}
Hence it is a decreasing function of the mortality rate. This expression also shows that with von Neumann--Morgenstern preferences, the finite fund buys systematic-mortality-risk insurance (i.e. the insurance purchase rate is positive) when it is less risk averse than the large fund, $\alpha_1>\alpha_2$, and sells systematic-mortality-risk insurance (i.e. the insurance purchase rate is negative) when it is more risk averse, $\alpha_1<\alpha_2$. The action under Epstein-Zin utility is not always clear, but for $\rho_1=\rho_2:=\rho$, we see the actions remain the same when $\rho<0$ and flip when $0<\rho<1$ i.e.\  insurance is bought when $\alpha_2>\alpha_1$ and sold for $\alpha_1>\alpha_2$.

\begin{figure}[ht]
    \centering
        \includegraphics[width=0.6\textwidth]{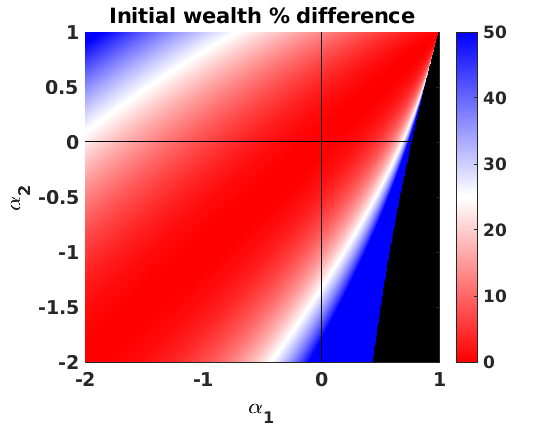}
        \caption{The insurance benefit
        under our stylised model when $a=4$, $b=1$ and $\mu=r=0$. This is the case of von Neumann--Morgenstern preferences where $\alpha_1$ denotes the preference of the finite fund and $\alpha_2$ the preference of the infinite fund. 
        Regions in black correspond to points where the analytic solution of the problem, or the associated one-fund problem (i.e.\ no insurance), is complex, so one or other of the problems is ill-posed.  The same shade of red is used whenever the maximal benefit is 10\% or less, and the same shade of blue is used whenever the maximal benefit is 50\% or greater.}
    \label{fig:differences_two_fund}
\end{figure}

\subsection{Cairns--Blake--Dowd mortality model}

We now consider a realistic mortality model to provide meaningful numbers on the differences in performance between a fund with and without insurance.
We consider a 1-factor continuous time analogue of the two-factor Cairns--Blake--Dowd model \cite{armstrong_buescu_dalby}, which leads to the following mortality rate equation:
\begin{multline}
    d\lambda=((e^{\lambda} - 1)(B_1 t(e^{\lambda} - 1) + B_1 t + B_2(B_3 t^2 +\\ B_4 t + 1) + B_5 e^{\lambda})/((e^{\lambda} - 1)^2 + 2 e^{\lambda} - 1)) dt \\
    + (B_6((e^\lambda - 1)^2 e^{-2\lambda} - (e^{\lambda} - 1)e^{-\lambda})(B_7(B_8t + 1)^2 + 1)^\frac{1}{2}e^\lambda) dW_t. \label{eq:CBD_lambda}  
\end{multline}
The constant coefficients $B_i$ can be found in \Cref{table:coefficients}. For completeness, details of the derivation of this SDE can be found in \Cref{sec:CBD_details}.
\begin{table}[ht]
\centering
\caption{Coefficients for the SDE \eqref{eq:CBD_lambda}.}
{\begin{tabular}{ll}
 \hline
     Coefficient & Value \\
    \hline
     $B_1$ &  0.00118 \\
     $B_2$ &  0.00317 \\
    $B_3$ &  $1.04\times 10^{-5}$ \\
    $B_4$ &  0.00125 \\
    $B_5$ &  0.0773 \\
    $B_6$ &  0.0782 \\
    $B_7$ & 0.0393 \\
    $B_8$ &  0.0166 \\
    \hline
\end{tabular}}
\label{table:coefficients}
\end{table}

For this mortality model, the PDEs to solve are \eqref{eq:inf_fund_optimal_hjb} and \eqref{eq:g_eqtn}, with $\drift_\lambda$ and $\vol_\lambda$ replaced by the drift and volatility in \eqref{eq:CBD_lambda}. Full details of our numerical approach for solving these equations is given in \Cref{sec:numerical_approach}. Our market parameters are $r=0.027$, $\mu=0.062$, $\sigma=0.15$. From the work in \cite{armstrong_buescu_dalby}, we expect meaningful non-trivial solutions to exist when $\alpha$ and $\rho$ have the same sign only. In \cite{xing}, values of $\alpha\in(-7, -1)$ and $\rho\in(-4,1/2)$ are considered. The arguments in \cite{armstrong_buescu_dalby} show $\rho=-1$ can be a reasonable choice and $\alpha\leq\rho$ should hold, otherwise the preferences are risk-seeking in the satisfaction. With this in mind, for $\alpha,\rho<0$, when working with Epstein--Zin utility, we fix $\rho=-1$ and consider $\alpha\in(-10,-1)$. The situation when $\alpha>0$ is not typically considered in the literature, but it is a qualitatively distinct and interesting case, 
so we consider this too. For $\alpha,\rho>0$, when working with Epstein--Zin utility, we fix $\rho=1/3$. We also use power utility for some results.

\medskip

In \Cref{table:two_fund}, we summarise our findings on the insurance benefit for certain preference combinations of the small and large fund under the CBD model. The general trend the table reveals, is that the larger the difference in the two funds risk aversion, the greater the benefit insurance may provide. We also find that whether the finite fund buys or sells insurance in our CBD model, is consistent with the expression \eqref{eq:style_rate} for our stylised model. Hence the fund buys insurance above the zero diagonal of the table, with the exception of the preference combination $\alpha_1=1/4,\alpha_2=3/20,\rho_1=\rho_2=1/3$, where it sells, and sells insurance below the diagonal, with the exception of the preference combination $\alpha_1=3/20,\alpha_2=1/4,\rho_1=\rho_2=1/3$, where it buys.

Provided the difference in the two funds risk aversion is not too large (i.e.\  $|\alpha_1-\alpha_2|\leq 3$ (excluding the case $\alpha_1=1/4,\rho_1=1/3, \alpha_2=-2,\rho_2=-1$)) the insurance benefit is less than 6\%. Since this also defines the maximal benefit that could be achieved if the two funds were finite, the true benefit will be less than this. So, in these situations, it may not be a worthwhile exercise for the two funds to exchange insurance given the additional complexity it adds to managing the fund for relatively small benefit. Otherwise, the benefit is typically less than 15\%, indicating these situations may benefit from insurance. However, there are clearly situations in the top right and bottom left of \Cref{table:two_fund} where insurance is likely to be beneficial. For instance, when $\alpha_1=1/4, \rho_1=1/3$ and $\alpha_2=-10,\rho=-1$, the maximal benefit comes out at 6196\%. Clearly in this case it would be valuable to now try to solve the full equilibrium problem to  obtain the true benefit in the finite case.

\begin{table}
\begin{tabularx}{\textwidth} { 
  | >{\raggedright\arraybackslash}X 
  | >{\centering\arraybackslash}X 
  | >{\centering\arraybackslash}X
  | >{\centering\arraybackslash}X
  | >{\centering\arraybackslash}X
  | >{\centering\arraybackslash}X
  | >{\raggedleft\arraybackslash}X | }
 \hline
   & $\alpha_1=-10,\newline \rho_1=-1$& $\alpha_1=-5,\newline \rho_1=-1$ & $\alpha_1=-3,\newline \rho_1=-1$ & $\alpha_1=-2,\newline \rho_1=-1$ & $\alpha_1=3/20, \newline\rho_1=1/3$& $\alpha_1=1/4,\newline \rho_1=1/3$\\
   \hline
 $\alpha_2=-10, \newline\rho_2=-1$ & 0\% & 7.76\% & 22.2\%  & 37.6\% & 623\% & 6196\%\\
 \hline
 $\alpha_2=-5, \newline\rho_2=-1$ & 4.96\% & 0\% & 1.93\% & 5.48\% & 47.4\% & 92\% \\
 \hline
 $\alpha_2=-3, \newline\rho_2=-1$ & 10.2\% & 1.41\% & 0\% & 0.62\% & 14\% &  22.7\% \\
 \hline
 $\alpha_2=-2, \newline\rho_2=-1$ & 13.6\% & 3.22\% & 0.5\% & 0\% & 5.72\% & 8.42\% \\
 \hline
 $\alpha_2=3/20$ \newline $\rho_2=1/3$ & 21.8\% & 8.89\% & 4.32\% & 2.29\% & 0\% & 0.065\% \\
\hline $\alpha_2=1/4$ \newline $\rho_2=1/3$ & 21.1\% & 8.32\%  & 3.87\%  & 1.93\% & 0.041\%   & 0\% \\
 \hline
\end{tabularx}
\caption{Insurance benefit for the finite-infinite fund problem under different preference combinations. 
Here we consider Epstein--Zin utility. The first row defines the preferences of the small fund and the first column the preferences of the large fund, so that the value in a given cell is the result of these two funds trading. 
}
\label{table:two_fund}
\end{table}


\begin{figure}[ht]
      \centering
      \begin{minipage}{0.45\textwidth}
          \centering
        \includegraphics[width=1.0\textwidth]{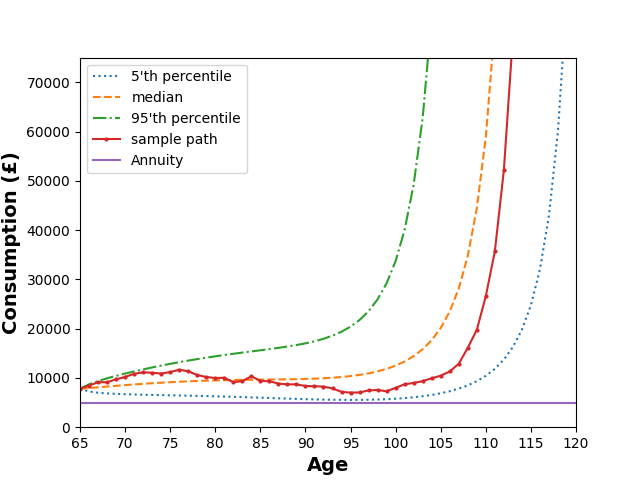}
      \end{minipage}
      \begin{minipage}{0.45\textwidth}
          \centering
        \includegraphics[width=1.0\textwidth]{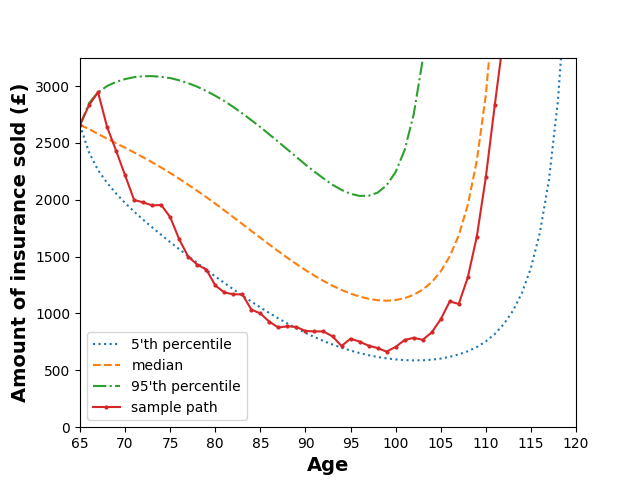}
      \end{minipage}
      \begin{minipage}{0.45\textwidth}
          \centering
        \includegraphics[width=1.0\textwidth]{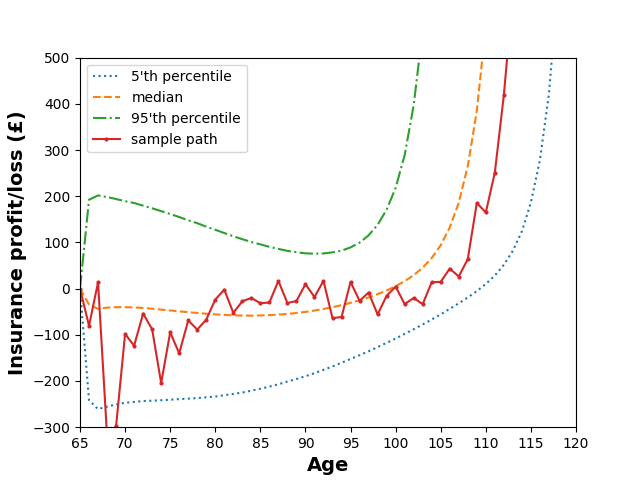}
      \end{minipage}
    \caption{Consumption, insurance purchase and insurance profit and loss fan diagrams (1,000,000 scenarios are used) for $\alpha_1=\rho_1=-3$ and $\alpha_2=\rho_2=-1$.}
    \label{fig:strats_p_s_-3_p_l_-1}
\end{figure}

\begin{figure}[ht]
      \centering
      \begin{minipage}{0.45\textwidth}
          \centering
        \includegraphics[width=1.0\textwidth]{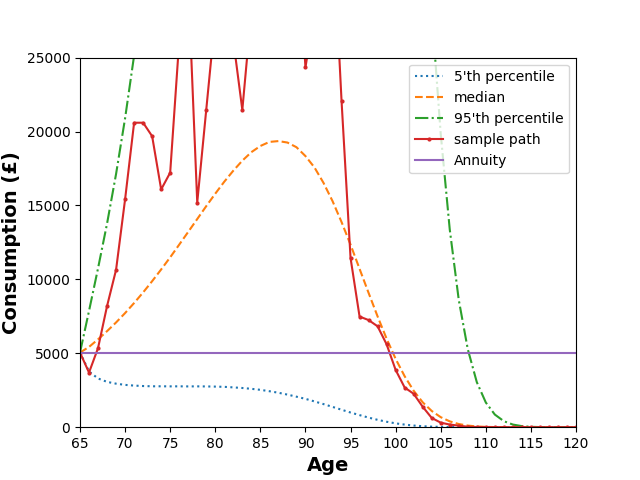}
      \end{minipage}
      \begin{minipage}{0.45\textwidth}
          \centering
        \includegraphics[width=1.0\textwidth]{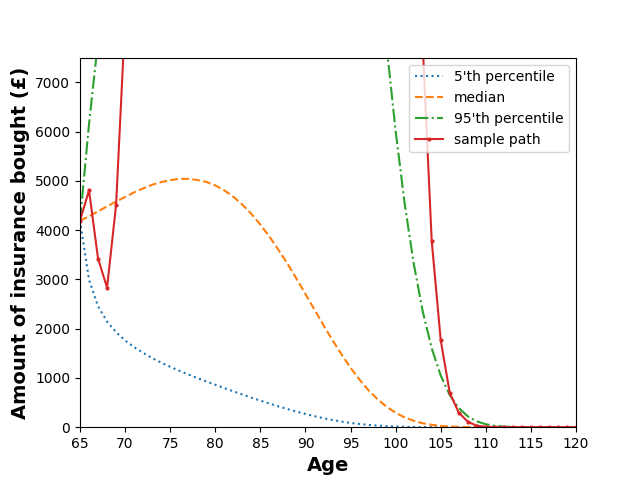}
      \end{minipage}
      \begin{minipage}{0.45\textwidth}
          \centering
        \includegraphics[width=1.0\textwidth]{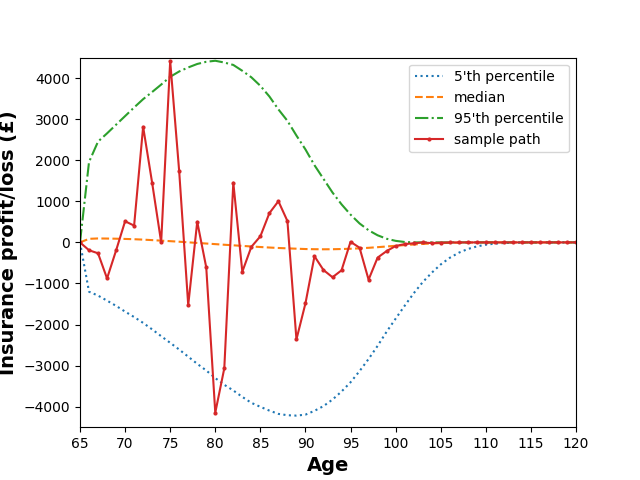}
      \end{minipage}
    \caption{Consumption, insurance purchase and insurance profit and loss fan diagrams (1,000,000 scenarios are used) for $\alpha_1=\rho_1=0.15$ and $\alpha_2=\rho_2=-1$.}
    \label{fig:strats_p_s_0.15_p_l_-1}
\end{figure}

\medskip
We now simulate some example pension outcomes for an individual with a pot size of £126,636, in an insured drawdown scheme with insurance contracts. This value is chosen to represent a median-earning female retiring in the UK in 2019. For comparison, in our results we include the annuity amount this pot could purchase, which under our CBD model is approximately £5000 per year.

In \Cref{fig:strats_p_s_-3_p_l_-1} and \Cref{fig:strats_p_s_0.15_p_l_-1}, we plot the consumption strategy, insurance purchase strategy and insurance profit and loss (each year), when the finite and infinite fund have power utility. We use power utility for computational ease. We now briefly describe the difference in consumption strategies for these funds when $\alpha>0$ and $\alpha<0$, in the absence of insurance contracts. See \cite{armstrong_buescu_dalby} for a detailed discussion. Because Epstein-Zin preferences (and power utility) are homogeneous of degree $\alpha$ with respect to wealth, they effectively target an infinite income when $\alpha<0$, whilst no income is admissible for $\alpha>0$. This leads to consumption strategies that put off consuming until approaching death when $\alpha<0$, so as to avoid running out of money, and strategies that are overly risk taking and consume early when $\alpha>0$, so as to avoid dying with leftover funds. We therefore interpret the case with $\alpha<0$ as modelling individuals/funds with inadequate pensions and $\alpha>0$ as modelling individuals/funds with adequate pensions. Consumption tends to infinity in later life for $\alpha<0$, while consumption is initially increasing and then decreases to zero when $\alpha>0$. This is a consequence of the homogeneity of Epstein--Zin preferences which prevents there being any form of hard or soft lower-bound on pension income.  \Cref{fig:strats_p_s_-3_p_l_-1} and \Cref{fig:strats_p_s_0.15_p_l_-1} show these comments remain true in the presence of insurance contracts.

In \Cref{fig:strats_p_s_-3_p_l_-1}, we plot the strategies and profit and loss for a finite fund with $\alpha_1=\rho_1=-3$, trading with an infinite fund with $\alpha_2=\rho_2=-1$. Solving the associated PDEs for this problem, the insurance benefit is 0.48\%. This is reflected in the average and median total consumption which improve slightly on the base scheme without insurance. We see even the 5-th percentile of scenarios outperforms an annuity in this case. 
Moving onto the insurance strategy, the small fund sells insurance to the large fund as it is more risk averse. 
As with the stylised model, the proportion of wealth spent on insurance decreases as the mortality rate increases i.e.\  the factor multiplying wealth in \eqref{eq:insurance_strat}. This is why the amount of insurance purchased is highest at the start of retirement and decreases until approximately age 100, at which point the explosion in wealth due to delayed consumption means the amount purchased increases significantly. Looking at the median scenario of the profit-and-loss plot, following this approach leads to a net loss on insurance up to age 105 (roughly), at which point significant profits begin to be made. For risk-averse individuals with inadequate pensions, this makes sense, as they are sacrificing consumption earlier in life to be insured that should they live too long, they will be covered. Hence, the insurance contract protects against systematic longevity risk. If we make the infinite fund more risk averse than the finite fund, the finite fund buys insurance, but the picture remains qualitatively unchanged.

In \Cref{fig:strats_p_s_0.15_p_l_-1}, we repeat the same process, but with a finite fund with $\alpha_1=\rho_1=0.15$. The insurance benefit in this case is 0.8\%, so again, higher average and median total consumption is achieved as expected. The median scenario outperforms an annuity in terms of total consumption, however, it provides no consumption from age 110 onwards. As we have already discussed, this arises from the positive homogeneity of Epstein--Zin preferences and is a limitation of using these preferences. The finite fund is less risk averse than the infinite fund so it buys insurance.
The proportion of wealth spent on insurance still decreases as the mortality rate increases, but since wealth can increase significantly from large investment in the risky asset, we see a maximum can be achieved, before all scenarios decrease to zero when the fund runs out of money. 
Interestingly, the median profit and loss scenario actually makes a loss, yet consumption is higher with insurance than without. This can be explained by the fact that insurance does yield a profit in the early stages of retirement (up to age 80 roughly), which means more wealth can be put into the risky asset, and the increased returns from this result in a net increase in utility over the funds lifetime. This approach seems consistent with a funds preference to consume as much as possible in early retirement when $\alpha>0$, to avoid dying with leftover money. If we make the infinite fund less risk averse than the finite fund, the finite fund now sells insurance, but the picture remains qualitatively unchanged.



\section{Conclusions}\label{sec:conclusions}

In this paper, we have shown how two collective funds can be operated to exchange mutual insurance contracts on systematic mortality risk in a fair and optimal manor. We have studied the specific case of an insured drawdown scheme characterised by a longevity credit system, however, the defining equation that facilitates the exchange of insurance, the market clearing condition, could be applied to any two collective funds. We consider the case where the two funds are exposed to the same mortality risk. In this situation one may intuitively expect there would be no reason for the two funds to trade insurance. However, we demonstrate this is not true and only a difference in the funds attitudes to risk is needed for beneficial insurance to be exchanged. It is likely that if the two funds face different mortality risk, the benefit of insurance will be greater and this is a question to address in future work.

We focus on a limiting case of the general equilibrium problem where one of the funds is finite and the other infinite. This is a far more tractable version of the fully general equilibrium problem for two finite funds, but still yields valuable insight as it defines the maximum benefit the two funds can achieve trading with each other.

Under a stylised mortality model, we solve the problem analytically. This confirms insurance contracts are always beneficial as would be expected (at least when a non-trivial solution is well defined). It reveals when a fund buys or sells insurance, and it provides a rough upper bound on how much benefit insurance can provide.
Using a more realistic Cairns--Blake--Dowd mortality model, we solve the problem numerically to show how much benefit insurance may give in practice. We see that the bigger the difference in the two funds risk aversion the greater the benefit insurance may provide. This is because there is a bigger gap between the two funds valuation of what the fair no trade price is and this can be exploited. When preferences are not too dissimilar, these prices are close and this means the maximum benefit is generally less than 6\%. As the true benefit for two finite funds is less than this, this suggests that the exchange of mutual insurance contracts is not a worthwhile exercise in such cases and the base insured-drawdown scheme is close to optimal.

In future work, we will generalise this approach of generating and defining insurance markets via market clearing conditions for other risk factors, such as wage growth risk. This will likely result in challenging systems of equations once more, but provided a no-trade price can be defined and the funds have different attitudes to risk, useful estimates on the maximum benefit insurance can provide may still be obtained.

\section*{Acknowledgements}

This research is funded by Nuffield grant FR-000024058.

\bibliographystyle{plain}
\bibliography{collectivization}

\appendix

\section{Justification for continuous-time aggregator}
\label{sec:ezAggregatorMotivation}

If an individual's time of death is independent of the systematic factors, their discrete-time Epstein-Zin utility with mortality
satisfies
\[
Z_{t}^\rho=c_t^\rho + e^{-(\delta + \frac{\rho}{\alpha} \lambda_t) \delta t} \E_{\P}( Z_{t+\delta t}^\alpha  \mid {\cal F}_t)^{\frac{\rho}{\alpha}}. 
\]
where we have introduced a discounting rate $\delta$ so that $\beta=e^{-\delta \, \delta t}$.
Rearranging we find
\[
\E_{\P}( Z_{t+\delta t}^\alpha  \mid {\cal F}_t) = \left[
\frac{Z_{t}^\rho- c_t^\rho}
{e^{-(\delta + \frac{\rho}{\alpha} \lambda_t) \delta t}}
\right]^\frac{\alpha}{\rho}. 
\]
We define $V$ by requiring $\alpha V=Z^\alpha$. The coefficient $\alpha$ is there to ensure that the transform $Z \to V$ is monotone so that $V$ is a gain
function that defines identical preferences to $Z$. So
\[
\E_\P( V_{t+\delta t} \mid {\cal F}_t )
= \frac{1}{\alpha}
\left[
\frac{
(\alpha V)^{\frac{\rho}{\alpha}} - \delta t \, c^\rho
}
{e^{-(\delta + \lambda \frac{\rho}{\alpha})\Delta t }}
\right]^{\frac{\alpha}{\rho}}
\]
Proceeding formally, we use l'H\^opital's rule to find:
\begin{align}
\frac{\ed}{\ed s} \E_\P( V_{s} \mid {\cal F}_t )
\Big|_{s=t}
&=\lim_{\delta t \to 0} \frac{\ed}{\ed (\delta t) }
\left[
\frac{1}{\alpha}
\left[
\frac{
(\alpha V)^{\frac{\rho}{\alpha}} - \delta t \, c^\rho
}
{e^{-(\delta + \lambda_t \frac{\rho}{\alpha})\delta t }}
\right]^{\frac{\alpha}{\rho}}
\right] \nonumber \\
&=\frac{1}{\alpha} \frac{\alpha}{\rho} (- c^\rho)((\alpha V)^\frac{\rho}{\alpha} )^{\frac{\alpha}{\rho}-1}
+ \left( \frac{\alpha}{\rho} \delta + \lambda_t \right) V
\nonumber \\
&=-\frac{1}{\rho}c^\rho \left( \alpha V \right)^{1-\frac{\rho}{\alpha}} + \left( \frac{\alpha}{\rho} \delta + \lambda_t \right) V.
\label{eq:ezAggregatorMotivation}
\end{align}
This then motivates the definition for the Epstein--Zin aggregator with mortality as a solution to the BSDE \eqref{eq:ezBSDE}
will satisfy equation \eqref{eq:ezAggregatorMotivation}.

\section{Continuous-time CBD mortality model}\label{sec:CBD_details}

We consider a continuous time version of the two-factor Cairns--Blake--Dowd model, for which the underlying mortality effects are described by:
\begin{align}
    dA := d(A_1,A_2)= \mu dt + C dW_t,
\end{align}
where 
\begin{align*}
    \mu&:=(\mu_1,\mu_2) = (-0.00669, 0.000590),\\ 
    V&=CC^T=\begin{pmatrix}
            0.00611 & -0.0000939\\ 
            -0.0000939 & 0.000001509
            \end{pmatrix},\\
    W_t&:=(W^1_t,W^2_t),
\end{align*}
for independent Brownian motions $W^1_t,W_t^2$
The parameter values used come from \cite{CBD_article} equation (4).
C is the upper triangular matrix from the Cholesky decomposition of V i.e.\  
\begin{equation}
    C=
    \begin{pmatrix}
            0.0782 & -0.00120\\ 
            0 & 0.000257            
            \end{pmatrix}.\label{eq:vol_matrix}
\end{equation}
The process $A_1$ captures general improvements in mortality over time at all ages (and therefore trends downwards with time), while $A_2$ captures the fact mortality  improvements have been greater at lower ages (and therefore trends upwards with time) \cite{CBD_article}. 

We next define $p(x_0,t)$ to be the survival probability at time $t$, for a cohort of age $x_0$ at time $t=0$, and $$q(x,t)=1-p(x,t)=\frac{e^{A_1 +A_2 (x_0+t)}}{1+e^{A_1 +A_2 (x_0+t)}}$$ to be the probability of death at time $t$. We assume $q(x_0,t)=1-e^{-\lambda(x_0,t)}$, with $\lambda(x_0,t)$ essentially being a time dependent rate parameter for an exponential distribution. We therefore interpret $\lambda(x_0,t)$ as our mortality rate and have 
\begin{equation}  
\lambda(x_0,t) = -\text{log}\left(1-\frac{e^{A_1 +A_2 (x_0+t)}}{1+e^{A_1 +A_2 (x_0+t)}}\right).\label{eq:rate_def}
\end{equation}

We would next like to derive the SDE for $\lambda$, which is done by an application of It\^{o}'s Lemma. In its current form, we obtain an SDE with a stochastic drift due to the $A_2 (x_0+t)$ term. This increases the dimension of the HJB equation from 2 to 3. For simplicity, we make the following modification. Note that the SDE for $A_1+A_2(x_0+t)$, is given by
\begin{align}
    d(A_1&+A_2(x_0+t)) \nonumber\\
    &= dA_1 + dA_2 (x_0+t) + A_2 dt\nonumber\\
    &= (\mu_1 +\mu_2(x_0+t) +A_2)dt+C_{1,1}dW^1_t+C_{1,2}dW^2_t+(x_0+t)C_{2,2}dW^2_t\nonumber\\
    &=(\mu_1 +\mu_2(x_0+t) +A_2)dt+\left(C_{1,1}^2+(C_{1,2}+(x_0+t)C_{2,2})^2\right)^\frac{1}{2} d\Tilde{W}_t,\label{eq:underlying_SDE}
\end{align}
where $C_{i,j}$ denotes the $i,j$'th component of \eqref{eq:vol_matrix} and $\Tilde{W}$ is a new independent Brownian motion. To obtain a tractable problem, we define a new variable $x:=A_1+A_2(x_0+t)$ in \eqref{eq:rate_def}, which satisfies the SDE
\begin{equation}
    dx = (\mu_1 +\mu_2(x_0+t) +A_{2,0}+\mu_2 t)dt+\left(C_{1,1}^2+(C_{1,2}+(x_0+t)C_{2,2})^2\right)^\frac{1}{2} d\Tilde{W}_t,\label{eq:simple}
\end{equation}
i.e.\ \eqref{eq:underlying_SDE} with $A_2$ replaced by the solution of the ODE $dA_{2}=\mu_2 dt$. This then yields a HJB equation dependent on $t$ and $\lambda$ only. Looking at Figure 1 in \cite{CBD_article}, we choose $A_2^0=0.1058$.

\section{Numerical approach for solving the CBD problem}\label{sec:numerical_approach}

To solve \eqref{eq:inf_fund_optimal_hjb} for the CBD model, we first need to solve \eqref{eq:g_eqtn} to obtain the value function of the large fund, since this dictates the price of insurance contracts. 
We use the Crank--Nicolson scheme, along with an improved Euler step to generate the first estimate of the solution at each time step.
We set $x_0=65$ to be the starting age of our simulations. $\lambda=0.01$ corresponds to a survival probability of 0.99 at retirement age, while $\lambda=10$ yields a survival probability of $4.5\times 10^{-5}$, so that all individuals are almost surely dead at this point. We are therefore interested in the solution for $\lambda\in[0.01,10]$.

The value of the solution at the boundaries is not obvious a priori. As such, we enforce Neumann boundary conditions at the ends of our domain and extend our computational domain to $\lambda\in[0.001,20]$, to minimise the impact of the boundary conditions on the solution for $\lambda\in[0.01,10]$. We expect $g_2\to\infty$ as $\lambda\to0$ and $g\to0$ as $\lambda\to\infty$. Because of this expected behaviour, we take advantage of two transformations: (i) we use the change of variables $L=\textrm{log}(\lambda)$ and (ii) solve the HJB equation for $\textrm{log}(g_2)$. We first solve \eqref{eq:g_eqtn} with $\vol_\lambda=0$, as we only require one boundary condition plus the payoff. Using the stylised model as a guide, we expect $g_2\approx \lambda^{\frac{\alpha_2 (\rho_2-1)}{\rho_2}}$ to be a rough approximation to the solution. Taking into account our transformations, we therefore enforce $$\frac{\partial}{\partial L}  (\textrm{log} (g_2))= \frac{\alpha_2(\rho_2-1)}{\rho_2}.$$ With the risk-less solution $g_{2,det}$ obtained, we solve the full HJB equation with $$\frac{\partial}{\partial L}  (\textrm{log} (g_2(20)))=\frac{\partial}{\partial L}  (\textrm{log} (g_{2,det}(20)))$$ and $$\frac{\partial}{\partial L}  (\textrm{log} (g_2(0.001)))=\frac{\partial}{\partial L}  (\textrm{log} (g_{2,det}(0.001))),$$ where the derivatives are approximated using backward and forward difference approximations respectively.

Again, informed by \eqref{eq:analytical} from the stylised model, we take our payoff at the final time $T_f$, to be $$g_{2,T_f}=\lambda^{\frac{\alpha_1 (\rho_1-1)}{\rho_1}},$$
which should be a rough approximation to the correct payoff. Since this is an approximation, we would like to minimise its impact on the solution. We therefore run our numerical scheme back from time $T_f=150$ (which corresponds to age 215). Our computational domain for $\lambda$ extends to $\lambda=20$, and in a deterministic setting (i.e.\ the volatility is zero in \eqref{eq:CBD_lambda}) the simulated value of $\lambda$ is approximately 20 at age 215, hence the choice. This time is far longer than is required for all individuals to die (our simulations suggest everyone will be dead by age 120) and therefore achieves our aim of minimising the payoff's impact.

With the solution of the infinite fund problem inputted to \eqref{eq:inf_fund_optimal_hjb}, the approach for solving this equation is similar to that just described. We again use the Crank-Nicholson scheme, solve the problem for $\lambda\in[0.001,20]$ with Neumann boundary conditions and compute the solution backwards in time from $T_f=140$, with payoff
\begin{equation}
    g_{1,Tf}=\lambda^{\frac{\alpha_1(\rho_1-1)}{\rho_1}}.
\end{equation}
The difference is how we define the Neumann boundary conditions. To do this, we now solve the one-fund problem \eqref{eq:g_eqtn}, but for the finite fund (i.e.\ replacing $(\alpha_2,\rho_2)$ with $(\alpha_1,\rho_1)$). With this solution obtained, we approximate the derivative at $\lambda=0.001$ using a forward approximation and the derivative at $\lambda=20$, using a backward approximation, these approximations are the Neumann boundary conditions we impose to solve \eqref{eq:inf_fund_optimal_hjb}.

We would also like to generate typical wealth and consumption streams for a fund that behaves optimally. This can be done as follows. Having obtained the numerical solution $g_{num}(\lambda)$ to \eqref{eq:g_eqtn}, it is not difficult to find constants $A,B$ such that $g_{num}\approx A \lambda^B$. This can then be used to approximate $c^*$ and $q^{c*}$ in \eqref{eq:finite_strats}. Along with 
the investment proportion from the Merton problem, this can be feed into \eqref{eq:wealth}. We can then simulate this approximation to \eqref{eq:wealth} using the Euler-Maruyama method \cite{intro_book}.

\end{document}